\pdfoutput=1
\documentclass[a4paper,10pt]{article}
\usepackage[natbibapa]{apacite}
\usepackage{authblk}
\usepackage[utf8]{inputenc}
\usepackage{amsmath,amsthm, amssymb, amsfonts, amsbsy}
\usepackage[mathscr]{euscript}
\usepackage{xspace}
\usepackage{graphicx}
\usepackage{subcaption}
\usepackage{color}
\usepackage[table]{xcolor}
\usepackage{enumerate}
\usepackage[english]{babel}
\usepackage{verbatim}
\usepackage{pgf, tikz}
\usepackage{pgfplotstable}
\usepackage{pgfplots}
\usepgfplotslibrary{colorbrewer}
\pgfplotsset{compat = 1.15, cycle list/Set1-8} 
\usetikzlibrary{pgfplots.statistics, pgfplots.colorbrewer} 
\usepackage{pgfplotstable}

\usepgfplotslibrary{statistics}
\usetikzlibrary{arrows, automata}
\usetikzlibrary{fit}
\usetikzlibrary{shapes.geometric}
\usepackage{ifthen}
\usetikzlibrary{calc}
\usepackage[inline, shortlabels]{enumitem}
\usepackage{float}
\usepackage{soul}
\usepackage{rotating}
\usepackage{booktabs} 
\usepackage{colortbl}
\usepackage{tabularx}
\usepackage{multicol}
\usepackage{url}

\definecolor{lightgray}{gray}{0.9}
\restylefloat{table}

\usepackage[draft]{fixme}
\fxsetup{
  layout=marginnote,
  marginface=\normalfont\tiny,
  envface=,
  inlineface=,
  innerlayout=noinline,
}

\usepackage[ruled, vlined, linesnumbered, resetcount]{algorithm2e}

\SetKw{Break}{break}
\SetKw{True}{true}
\SetKw{False}{false}
\DontPrintSemicolon

\DeclareMathOperator*{\val}{val}

\newcommand{\R}{\mathbb{R}}
\newcommand{\B}{\mathbb{B}}
\newcommand{\Z}{\mathbb{Z}}

\newcommand{\NP}{\textsf{NP}}
\newcommand{\Pclass}{\textsf{P}}

 \let\mathscr\relax \usepackage[scr]{rsfso}

\newcommand{\stab}{\textsc{stab}}
\newcommand{\poly}{\mathcal{P}}
\newcommand{\polyR}{\mathcal{R}}
\DeclareMathOperator{\conv}{conv}

\newcommand{\pmc}{\textsc{pmc}}
\newcommand{\vcp}{\textsc{vcp}}

\newcommand{\antiweb}{\bar{W}^{q}_{\ell}}
\newcommand{\web}{W^{q}_{\ell}}

\theoremstyle{plain}
\newtheorem{theorem}{Theorem}
\newtheorem{proposition}[theorem]{Proposition}
\newtheorem{observation}[theorem]{Observation}
\newtheorem{corollary}[theorem]{Corollary}
\newtheorem{lemma}[theorem]{Lemma}

\theoremstyle{definition}

\begin{document}
\title{Compact formulations and valid inequalities for parallel machine scheduling with conflicts}

\author{Phablo F. S. Moura, Roel Leus, and Hande Yaman}
\affil{Research Center for Operations Research \& Statistics, KU Leuven, Belgium\\
\texttt{\normalsize\{phablo.moura, roel.leus,  hande.yaman\}@kuleuven.be}
}

\maketitle

\begin{abstract}
\noindent The problem of scheduling conflicting jobs on parallel machines consists in assigning a set of jobs 
to a set of machines so that no two conflicting jobs are allocated to the same machine, and the maximum processing time among all machines is minimized.
We propose a new compact mixed integer linear 
formulation 
based on the representatives model for the vertex coloring problem, which overcomes a number of issues inherent in 
the natural assignment model.  
We present a polyhedral study of the associated polytope, and describe classes of valid inequalities inherited from the stable set polytope.
We describe branch-and-cut algorithms for the problem, and  report on computational experiments with benchmark instances. 
Our computational results on the hardest instances of the benchmark set show that the proposed algorithms are superior (either in running time or quality of the solutions) to the current state-of-the-art methods. 
We find that our new method performs better than the existing ones especially when 
the gap between the optimal value and the trivial lower bound (i.e., the sum of all processing times divided by the number of machines) increases.


\end{abstract}

\section{Introduction}
Given an undirected graph $G=(V,E)$ (a.k.a.\ \emph{conflict graph}) where $V$ is a set of $n$ vertices (representing the jobs),  processing times $p \colon V \to \Z_>$, and $m\geq 2$ identical machines, the problem \textsc{Parallel Machine Scheduling with Conflicts} (\pmc) consists in finding an assignment $c \colon V \to \{1,\ldots, m\}$ with $c(u)\neq c(v)$ for all $\{u,v\} \in E$ that minimizes $\max_{k \in \{1,\ldots,m\}} \sum_{v \in V \colon c(v)=k} p(v)$, that is, the makespan.
This problem is clearly $\NP$-hard as it contains \mbox{$3$-\textsc{partition}}.

This problem was first studied by~\cite{BODLAENDER1994219}, who designed polynomial-time approximation schemes (PTASs) for the cases where the conflict graph is bipartite, complete multipartite, or has bounded treewidth.
Very recently, \cite{FURMANCZYK2024106606} showed a 2-approximation algorithm for \pmc, and a PTAS for the case when the jobs have processing times of unit duration.
A number of solving methods based on mixed-integer linear programming (MILP)  have also been proposed for \pmc\ in the literature.
One is the branch-and-price algorithm due to~\cite{BIANCHESSI2021105464}, which uses a MILP formulation with a binary variable for each stable set of the conflict graph and each machine, and a non-negative real variable for every machine.
Another algorithm (the first-published computational study, to our knowledge) is a binary search using a set-covering model for Bin Packing with Conflicts (\textsc{bppc}) devised by~\cite{kowalczyk2017exact}.
Given $m \in \Z_>$ identical bins with capacity $C \in \Z_>$, a set~$V$ of~$n$ items with capacity consumption $p \colon V \to \Z_>$, and a (conflict) graph $G=(V,E)$, \textsc{bppc} consists in finding an assignment of items to a minimum number of bins so that the capacity of the bins is respected and no pair of items in $E$ are assigned to the same bin.
Note that the decision versions of \pmc\ and \textsc{bppc} are equivalent: an instance $(G,p,m)$ of \pmc\ has a solution of makespan~$C$ if and only if the items in $V(G)$ can be assigned (without conflicts) to $m$ bins of capacity $C$.
\cite{Sadykov13} proposed a branch-and-price approach to~\textsc{bppc}, and~\cite{Epstein08} designed constant-factor approximation algorithms for the problem restricted to bipartite graphs and perfect graphs.
\cite{MALLEK2019357} proposed heuristics and MILP formulations for a variant of \pmc\ where the machines run at different speeds and all jobs are unit time.
This problem was shown to be $\NP$-hard even when restricted to instances with two machines and a forest as the conflict graph~\citep{mallek2024scheduling}.

Let~$[\ell]$ denote the set $\{1,\dots,\ell\}$ for any $\ell \in \Z_>$.
Scheduling with conflicts is closely related to the vertex coloring problem (\vcp), which consists in, given a graph $G=(V,E)$,  finding a coloring $c \colon V \to [\ell]$ with $\ell \in \Z_>$ colors such that $c(u)\neq c(v)$ for all $\{u,v\} \in E$, and $\ell$ is minimum.
The smallest $\ell \in \Z_>$ such that $G$ admits an $\ell$-coloring is denoted by  $\chi(G)$.
Obviously, any feasible solution for an instance of  \pmc\ with conflict graph $G$ and $m$ machines is a (proper) vertex coloring of $G$ with $m$ colors.
Thus an instance $(G,p,m)$ of this problem is infeasible if and only if $m < \chi(G)$.
\vcp\ is a classical optimization problem with a vast literature devoted to algorithms, complexity, structural aspects, and computational experiments 
\cite[see, e.g.,][]{malaguti2010survey,tuza1997graph}.

It is therefore natural to use the knowledge about~\vcp\ to study~\pmc.
In this work, we focus on mixed integer linear programming (MILP) formulations, strong valid inequalities, and separation algorithms.

First, in Section~\ref{sec:assign}, we consider the natural assignment formulation for \pmc\ using binary variables indexed by the jobs and machines, discuss how to reduce symmetries in this model, and show that the optimal value of its linear relaxation is precisely the average makespan, that is, $\frac{1}{m} \sum_{v \in V} p(v)$, a trivial lower bound for \pmc.

In Section~\ref{sec:rep-form}, we propose a compact MILP formulation for \pmc\ to alleviate some issues related to symmetry and unbalancedness of the enumeration tree associated with the assignment model.  
The proposed formulation for \pmc\ uses a set of representative jobs (one for each machine)  to represent feasible solutions to the problem, and is based on the asymmetric representatives model for \vcp\ introduced by~\cite{CamCamCor08}. 
A similar idea of representatives is used to model the classical one-dimensional bin packing problem in~\cite{HadjSalem20}.
In Section~\ref{sec:rep-poly} we conduct a polyhedral study of the associated polytope and establish some classes of valid inequalities inherited from the stable set polytope that are induced by specific classes of subgraphs.
We also briefly discuss the separation problems of some of these classes of inequalities.

Section~\ref{sec:experiments} presents a description of the implementation details of branch-and-cut algorithms for \pmc\, and a report on computational experiments using benchmark instances from the literature.
Our proposed solution method based on the (compact) representatives formulation is easier to implement than the best algorithms currently known for \pmc, which are all based on models with exponentially many variables and $\NP$-hard pricing problems.
Furthermore, the proposed method for \pmc\ is more amenable to enhancements, such as the inclusion of new cuts derived for the stable set and vertex coloring problems.

The computational experiments on the hardest instances in the benchmark instance set for \textsc{pmc} show that the proposed algorithms are in general superior (either in running time or quality of the solutions) to the current state-of-the-art methods.
Considering the entire set of instances in the benchmark, the branch-and-cut approach produces essentially the same average optimality gaps as obtained by the other algorithms in the literature.
Looking only at the instances with positive gap, our average gaps are up to 10 times smaller for the instances with a graphs density of at least $0.3$.
We also observe from our computational results that most of the (feasible) solved instances in the benchmark dataset with at least 50 vertices  have their optimal value very close to (at most 5\% larger than) the trivial lower bound~$\lceil\sum_{v \in V} p(v)/m\rceil$ for \pmc. This implies that the conflicts in these instances do not really affect the optimal value very much compared to the optimal makespan when scheduling the same jobs without conflicts.
To better understand the impact of the conflicts, we construct extra instances where the difference between the optimal value and the trivial lower bound can be large, and present empirical evidence that our method is more likely to outperform the best algorithm in the literature as this difference increases.

\section{Assignment formulation} \label{sec:assign}

A very natural way to model the set of feasible solutions for an instance $(G=(V,E),p,m)$ of \pmc\  is the one based on the assignment model for \vcp.
In such formulation, for each $v \in V$ and $k \in [m]$, there is a binary variable $\tilde x _{vk}$ indicating the machine that processes the job: $\tilde x_{vk} =1$ if, and only if, $v$ is assigned to machine $k$.
Moreover, the model contains a non-negative real variable~$\tilde y$ whose value is at least the maximum processing time over all machines~$k \in [m]$.
The following assignment formulation, denoted by (AF), for \pmc\ was introduced by~\cite{kowalczyk2017exact}.
\begin{align} 
     &\text{(AF)}& \min & ~ \tilde y  \nonumber \\ 
&&  \text{s.t.\;} & \sum_{k \in [m]} \tilde x_{vk} \geq 1, && \!\!\! \forall v \in V, \label{eq:assign}\\ 
    &&  & \tilde x_{uk} + \tilde x_{vk} \le 1 && \!\!\! \forall \{u,v\} \in E, k\in [m]. \label{eq:ass:edge}\\
    && & \sum_{v \in V} p(v) \tilde x_{vk} \leq \tilde y && \!\!\!\forall k \in [m],\label{eq:time}\\  
    &&  & \tilde x_{vk} \in \{0,1\} && \!\!\!\forall v \in V,  k\in [m], \\
    &&  & \tilde y\geq 0. && 
\end{align}
Constraints~\eqref{eq:assign} guarantee that each job is assigned to at least one machine, while constraints~\eqref{eq:ass:edge} imply that no conflicting jobs are assigned to the same machine.
The objective function together with constraints~\eqref{eq:time} establish the minimum makespan objective. 
This assignment formulation for \pmc\ has precisely $nm$+1 variables and $n+(|E|+1)m$ constraints, and so it is compact. 

Albeit simple and easily implemented, this model has some potential drawbacks.
First, the optimal value of its linear relaxation equals the average processing time per machine, a trivial lower bound for \pmc, which can be seen as follows.  Let $\underline{\tilde y}$ denote the optimal value of the linear relaxation of \rm{(AF)} on input $(G=(V,E),p,m)$ of \pmc.
First observe that the vector $\tilde x \in \R^{nm}$ with $\tilde x_{vk}=1/m$ for all $v \in V$ and $k \in [m]$ is a  feasible  solution (for any graph) for the linear relaxation of (AF).
    Thus, the optimal value of such relaxation is at most $\frac{1}{m} \sum_{v \in V}p(v)$.
    Consider now any feasible solution $(\underline{ \tilde y}, \tilde x)$ of the linear relaxation of (AF), and note that $\tilde x$  satisfies~\eqref{eq:assign}  for each~$v \in V$.
    By summing~\eqref{eq:time} over all $k \in [m]$, we have
    \[\underline{\tilde y} \geq \frac{1}{m} \sum_{v \in V}   p(v) \sum_{k\in [m]} \tilde x_{vk} \geq \frac{1}{m} \sum_{v \in V} p(v), \]
    where the last inequality holds because $\tilde x$ satisfies~\eqref{eq:assign}. 
    Therefore, the following observation holds.

\begin{observation}\label{obs:lp-value}
    It holds that \(\underline{\tilde y} = \frac{1}{m} \sum\limits_{v \in V} p(v).\)
\end{observation}

Another weakness of (AF) is that the branching strategy where a variable of fractional value is set to 1 in one subproblem and to 0 in the other subproblem tends to produce a very unbalanced enumeration tree. 
Because of \eqref{eq:assign}, setting $\tilde x_{vk}=1$ forces $\tilde x_{v\ell}=0$ for all $ \ell \in [m]\setminus \{k\}$, while $\tilde x_{vk}=0$ only prevents job~$v$ from being assigned to machine~$k$.
This was noticed by~\cite{MENDEZDIAZ2006826}  in the context of a branch-and-cut algorithm for \vcp.
Moreover, (AF) is symmetric, in the sense that 
any permutation of the machines yields a different solution with the same value. 
One may remove some of these symmetries by further imposing an ordering of the machines according to their processing times:
\begin{align}
     \sum_{v\in V}p\!\left( v\right) \tilde x_{vk}\geq \sum_{v\in V}p\!\left( v\right) \tilde x_{v,k+1} && \forall k \in\left[ m-1\right]. \label{ineq:sbc:time}   
\end{align}
Observe that with \eqref{ineq:sbc:time}  it is still possible to obtain equivalent solutions since machines that have the same total processing time are interchangeable.

To avoid such symmetries, one may use the constraints introduced by~\cite{MENDEZDIAZ2008159} to sort the stable sets in a vertex coloring by the minimum label of the vertices in each set, and consider only colorings that assign color $k$ to the  $k$-th stable set.
In the context of scheduling, assuming $V=[n]$, this means that a job $v$ can only be assigned to a machine $k\leq v$. 
These constraints are modeled as follows.
\begin{align}
     \tilde x_{vk}   &= 0& \forall v \in [m-1], k \in \{v+1, \dots, m\},\\
     \tilde x_{vk}   &\le \sum_{u=k-1}^{v-1} \tilde x_{u,k-1}& \forall v \in V\setminus \{1\}, k \in \{2,\ldots, \min(v,m)\}.\label{ineq:sbc:label}
\end{align}
Because of~\eqref{eq:assign}, the following inequalities are valid and strengthen~\eqref{ineq:sbc:label}:
\begin{align}
    \sum_{i=k}^{\min\{v,m\}} \tilde x_{vi}   &\le \sum_{u=k-1}^{v-1} \tilde x_{u,k-1}& \forall v \in V\setminus \{1\}, k \in \{2,\ldots, \min(v,m)\}.\label{ineq:sbc:label2}
\end{align}

\section{Representatives formulation}\label{sec:rep-form}

A vertex coloring of graph $G=(V,E)$ can be seen as a partition $\mathcal{S}$ of $V$ into $\ell \ge \chi(G)$ stable sets of $G$.
Each stable set $W \in \mathcal{S}$  is called a \emph{color class} and corresponds to the vertices that are assigned a same color in the coloring.
Given a nonempty color class $W\in \mathcal{S}$,
one can choose any vertex $v \in W$ to be its \emph{representative}.
For every vertex $u \in W\setminus \{v\}$,  $u$ is said to be \emph{represented} by $v$ (or that $v$ \emph{represents}~$u$).

The idea of representatives for the color classes in a vertex coloring was introduced by~\cite{CAMPELO2004159} to model the classical Vertex Coloring Problem.
Their integer linear programming model has binary variables $x_{vv}$ for each $v \in V$, and $x_{uv}$ for each $\{u,v\} \notin E$ with the following interpretation: $x_{vv}=1$ iff $v$ is the representative of its color class, and~$x_{uv}=1$ iff $u$ is a representative and $u$ represents $v$.

Consider an instance $(G,p,m)$ of \textsc{pmc} with $G=(V,E)$. 
The complement of~$G$, denoted by $\bar G$, is the graph $(V,\bar E)$ where 
$\bar E = \{\{u,v\} \subset V :\{u,v\} \notin E,  u\neq v\}$
is the complement of $E$.
We denote by $\bar e$ the size of $\bar E$.
Let $\prec$ be an arbitrary ordering on $V$.
We remark  that this ordering is not related to the order in which the jobs are processed (the latter ordering is irrelevant for the makespan objective).
As we shall discuss later on, the ordering $\prec$ avoids symmetrical solutions and  affects the dimension of polytope associated with the representatives formulation of the problem.
Before showing the model, we present some additional useful notation.

Let $\bar N^-(v)=\{u \in V : u\prec v, \{u,v\}\notin E\}$ and 
$\bar N^+(v)=\{u \in V : v \prec u, \{u,v\}\notin  E\}$ be the negative and positive anti-neighborhood of~$v$, respectively.  
Define $\bar N[v] = \bar N^+(v) \cup \bar N^-(v) \cup \{v\}$, $\bar N^{-}[v] =\bar N^-(v)\cup \{v\}$, and $\bar N^{+}[v] =\bar N^+(v)\cup \{v\}$. 
For each $v\in V$ and $u\in \bar N^{-}[v]$, 
define a variable $x_{uv}\in \{0,1\}$ that equals~$1$ if, and only if, the machine represented by job~$u$  contains job~$v$.
Additionally, there is a real variable $y$ which is an upper bound for the processing time of each machine.
We next introduce a representatives formulation for \textsc{pmc} on input $(G,p,m)$ with $G=(V,E)$.

\begin{align} 
     && \min & ~y  \label{eq:obj} \\
    && \text{s.t.\;} & \sum \limits_{v \in V} x_{vv} \leq m && \!\!\! \label{eq:machines}\\ 
    &&  & \sum \limits_{u \in \bar{N}^{-}[v]} x_{uv} \ge 1 && \!\!\!\forall v \in V, \label{eq:cover}\\ 
    &&  & x_{vu} + x_{vw} \leq x_{vv} && \!\!\!\forall v\in V, u,w \in \bar{N}^{+}(v) \text{ with } \{u,w\} \in E, \label{eq:edge}\\
    &&  & x_{vu} \leq x_{vv} &&  \!\!\!\forall v\in V,  u \in \bar{N}^{+}(v), \label{eq:link}\\
    && & \sum \limits_{u \in \bar{N}^{+}[v]} p(u) x_{vu} \leq y && \!\!\!\forall v \in V, \label{eq:processing-time}\\ 
    &&  & x_{uv} \in \{0,1\} && \!\!\!\forall v \in V,  u \in \bar{N}^{-}[v], \\
    &&  & y\geq 0. && \label{eq:non-neg}
\end{align}
The objective function~\eqref{eq:obj} together with constraints~\eqref{eq:processing-time} minimizes the maximum processing time of the machines.
A solution that uses at most~$m$ machines is guaranteed by constraint~\eqref{eq:machines}.
Constraints~\eqref{eq:cover} ensure that each job~$v$ is assigned to one machine, 
considering machines represented by~$v$ or by any of its negative anti-neighbors in $G$. 
Constraints~\eqref{eq:edge} guarantee that no conflicting jobs are assigned to the same machine.
Moreover, if a job is not a representative, it cannot represent any job according to constraints~\eqref{eq:link}.

Let us define the set of \emph{sources} as $S=\{v \in V :  \bar N^-(v)=\emptyset\}$ and \emph{sinks} as $T=\{v \in V :  \bar N^+(v)=\emptyset\}$.
Clearly, every source $v \in S$ can only be represented by itself, and so $x_{vv}=1$ holds for all feasible solutions. 
Thus, for each~$v \in S$, we simply substitute the variable~$x_{vv}$ with 1 and remove it from the model.
Then, inequality~\eqref{eq:machines} is replaced by
\begin{align}
    & \sum \limits_{v \in V\setminus S} x_{vv} \leq m -|S|. && \label{eq:machines2}
\end{align}
Moreover, inequalities~\eqref{eq:cover} are substituted by 
\begin{align}
     & \sum \limits_{u \in \bar{N}^{-}[v]} x_{uv} \geq 1 && \!\!\!\forall v \in V\setminus S. \label{eq:cover2}
\end{align}

Note that, for each  $v \in V\setminus T$ and $u,w \in \bar N^+(v)$ with $\{u,w\} \in E$, inequalities $x_{vu} \leq x_{vv}$ and $x_{vw} \leq x_{vv}$ are dominated by $x_{vu}+x_{vw} \leq x_{vv}$ if $v \notin S$, and $x_{vu}\leq 1$   and $x_{vw} \leq 1$ are dominated by $x_{vu}+x_{vw} \leq 1$ otherwise.
Therefore, we substitute inequalities~\eqref{eq:edge} and \eqref{eq:link} by the following ones.
\begin{align}
     & \sum_{u \in K} x_{vu} \leq x_{vv} && \!\!\!\forall v \in (V\setminus T) \setminus S, K \in \mathcal{K}(v), \label{eq:edge2:non-sources} \\
     & \sum_{u \in K} x_{vu} \leq 1 && \!\!\!\forall v \in (V\setminus T)\cap S, K \in \mathcal{K}(v), \label{eq:edge2:sources}
 \end{align}
where $\mathcal{K}(v)$ is the collection of all sets of vertices in $\bar N^+(v)$ that induce cliques of size $2$, or  maximal cliques of size $1$ in $G[\bar N^+(v)]$.

Finally, we replace inequalities~\eqref{eq:processing-time} by 
\begin{align}
    &p(v) x_{vv} + \sum \limits_{u \in \bar{N}^{+}(v)} p(u) x_{vu} \leq y && \!\!\!\forall v \in V\setminus S,\label{eq:processing-time2:non-sources}\\
    &p(v) \phantom{x_{vv}} + \sum \limits_{u \in \bar{N}^{+}(v)} p(u) x_{vu} \leq y && \!\!\!\forall v \in S.\label{eq:processing-time2:sources}
\end{align}

The representatives formulation \rm{(RF)} for \pmc\ is
\[\min\left\{y : (x,y) \in \{0,1\}^{n+\bar e- |S|} \times \R \text{ satisfies } \eqref{eq:machines2} - \eqref{eq:processing-time2:sources}\right\}.\]

Although the formulation correctly models \pmc\ for any ordering $\prec$ on the vertices, one can get additional properties when assuming certain orderings.
As an example, consider an ordering $\prec$ on $V$ such that there is $u \in S$ with  $p(u)\geq p(v)$ for all $v \in V$.
Since $u$ is a source, inequality~\eqref{eq:processing-time2:sources} guarantees that the optimal value of the linear relaxation of \rm{(RF)} is at least the maximum processing time, a trivial lower bound for the optimal value  of \pmc.

With respect to the lower bounds given by the linear relaxations of both (AF) and (RF), we next show that, in general,  no formulation is superior to the other.

\begin{proposition}\label{prop:lp-comparison}
Let  $\underline{\tilde  y}$ and $\underline{y}$ denote the optimal values
    of the linear relaxations of {\rm{(AF)}} and {\rm{(RF)}}, respectively, on instance $(G,p,m)$ of \pmc.
    It holds that $\underline{\tilde  y} - \underline{y}$ is not bounded.
\end{proposition}
\begin{proof}
First consider the family of instances where $G$ has~$n\ge 2$ vertices and no edges, $m=n$, $p(u)=n$ for a given $u \in V$, $p(v)=1$ for all $v \in V\setminus\{u\}$, and any ordering $\prec$ of the vertices of $G$ where $u$ is a source.
    For such instances, it is clear that $\underline{\tilde  y} = (2n-1)/n < 2$ and $\underline{y} \geq n$.

    Consider now a graph $G$ with $n\geq 4$ vertices $V=\{0,\ldots,n-1\}$ and no edges,  $m=2$, and $p(v)=1$ for all $v \in V$.
    Let us assume the ordering $i \prec i+1$ for every $i \in [n-2]$.
    Note that vertex 0 is the only source in $\prec$ (i.e. $S = \{0\}$).
    Consider now a vector $\bar x \in [0,1]^{n+\bar e -1}$ with $\bar e = n(n-1)/2$  such that 
    \begin{enumerate}
        \item $\bar x_{ii}=1/2^i$ for all $i \in V\setminus\{0\}$,
        \item $\bar x_{ij} = \bar x_{ii} =1/2^i$ for all $i,j \in V\setminus\{0\}$ with  $i<j$,
        \item $\bar x_{0j} = 1 - \sum_{i=1}^j 1/2^i = 1/2^j$ for all $j \in V\setminus\{0\}$.
    \end{enumerate}
     Note that $\sum_{i=1}^{n-1} \bar x_{ii} = \sum_{i=1}^{n-1} 1/2^i < 1$, and so $\bar x$ satisfies~\eqref{eq:machines2}.
    Since $G$ has no edges and $\sum_{i=0}^j \bar x_{ij}  = 1$ for all $j\in V\setminus\{0\}$, inequalities~\eqref{eq:cover2},~\eqref{eq:edge2:non-sources}, and~\eqref{eq:edge2:sources} hold for $\bar x$.
    Moreover, the value of~$\bar x$ is $\sum_{i=0}^{n-1} 1/2^i < 2$, and thus $\underline{y} < 2$. 
    On the other hand, $ \underline{\tilde  y} = n/2$ again by Observation~\ref{obs:lp-value}.
    Therefore, we conclude that $ \underline{\tilde  y} - \underline{y}$ is not bounded.
\end{proof}

Note that in the previous proof  all instances contain empty graphs (i.e. graphs with no edges), and so the proposition holds for the version without conflicts.
Additionally, the proof implies that it may be the case that every possible choice for the vertex ordering $\prec$ yields a representative formulation whose linear relaxation is weaker than the one given by the assignment formulation.
We also observe that if $n<m$, then constraint~\eqref{eq:machines2} of~(RF) is trivially satisfied and the value of its linear relaxation is at least the total processing time divided by~$n$, and this is strictly larger than the linear relaxation of~(AF).

We remark that the relation between the optimal values of the linear relaxations for \pmc\ presented in Proposition~\ref{prop:lp-comparison} is not observed when comparing the assignment and representatives models for the vertex coloring problem (\vcp).
In fact, one may prove that the optimal value of the linear relaxation of the representatives formulation for \vcp\ is at least as large as the optimal value of the linear relaxation of the assignment formulation for \vcp. 
To prove this, first assume (w.l.o.g.) that $V = \{1, \ldots, n\}$, and observe that the assignment model for \vcp\ has a binary variable $\tilde y_{k}$, for every $k \in V$, identifying whether the color~$k$ is used in the coloring or not.
Finally, given a solution $x$ of the linear relaxation of the representative model for \vcp, one may define a feasible  solution of same value in the linear relaxation of the assignment model for \vcp\ as follows: $\tilde y_k = x_{kk}$  for all $k \in V$, and  $\tilde x_{vk} = x_{kv}$ for all $v \in V, k \in \bar N^-[v]$.

The difference between the linear relaxation of the models for \vcp\ and \pmc\ shown above is due to the distinct objective functions of these problems: one has to minimize the number of distinct colors in the former, and the makespan in the latter.

\section{Representatives polytope} \label{sec:rep-poly}
In this section, we define a relaxation of the polytope associated with the representatives formulation, and present a polyhedral study of such polytope.
Our goal is to take advantage of the polyhedral results in the literature of \vcp\ to obtain valid (strong) inequalities for \pmc.

First recall that there is no solution satisfying all constraints of the formulation if $m < \chi(G)$.
If there is no coloring of $G$ with $m$ colors such that a vertex $v \in V$ induces a singleton color class $\{v\}$ (i.e., job $v$ appears alone in a machine), then every feasible solution satisfies $x_{vv}=0$ if $\bar N^+(v)=\emptyset$.
Hence, even when $m\geq \chi(G)$, the dimension of the polytope associated with the proposed formulation may depend on the existence of certain colorings of $G$.

To avoid making additional assumptions on the existence of certain types of colorings of $G$ and on the ordering of its vertices, we focus on studying the facial structure of a relaxed polytope $\poly(G,p) \subseteq \R \times \R^{n+\bar e-|S|}$ defined as 
$\conv \{(y,x) \in \R_\ge \times \{0,1\}^{n+\bar e - |S|} : (y,x) \text{ satisfies } \eqref{eq:cover2} -\eqref{eq:processing-time2:sources}\}$.
In other words, $\poly(G,p)$ is the relaxation of the polytope associated with (RF) that is obtained by removing capacity constraint~\eqref{eq:machines2}. 
Clearly,  every valid inequality for this polytope is also valid for the representatives polytope of \pmc.
We next explore the relationship between \vcp\ and \pmc\ to study the facial structure of~$\poly(G,p)$.

\cite{CamCamCor08} provided a detailed polyhedral study of the polytope associated with a formulation for the vertex coloring problem based on representatives.
Let us denote such polytope by $\polyR(G)$, and observe that inequalities \eqref{eq:cover2}--\eqref{eq:edge2:sources} are precisely the same constraints of their formulation.
As a consequence, it holds that 
$\polyR(G) =\conv\{x \in \{0,1\}^{n + \bar e -|S|} : x \text{ satisfies } \eqref{eq:cover2} - \eqref{eq:edge2:sources}\}$.

We next prove a simple lemma that allows us to obtain valid inequalities for~$\poly(G,p)$ from valid inequalities for $\polyR(G)$ in such a way that the dimension of the corresponding face is increased.
Such an increment in the dimension is desirable since (RF) has a variable $y$ that does not exist in the representatives model for \vcp.
In what follows,  an inequality $\pi x \leq \pi_0$ with $x \in \R^\ell$ is denoted by $(\pi,\pi_0) \in \R^\ell \times \R$.
Similarly, an inequality $\delta y + \pi x \leq \pi_0$ with $y \in \R$ and~$x \in \R^\ell$ is denoted by $(\delta, \pi,\pi_0) \in \R \times \R^\ell \times \R$.

\begin{lemma}\label{lem:lifting}
    Let $(\pi,\pi_0) \in \R^{n+\bar e-|S|}\times \R$ 
be a valid inequality for $\polyR(G)$, and let $F=\{x \in \polyR(G) : \pi x = \pi_0\}$ be its induced face.
    It holds that $(0,\pi,\pi_0) \in \R \times \R^{n+\bar e-|S|}\times \R$ is a valid inequality for $\poly(G,p)$, and $\dim(F') \geq \dim(F)+1$ where $F'=\{(y,x) \in \poly(G,p) : \pi x = \pi_0\}$. 
\end{lemma}
\begin{proof}
    It is clear that~$(0,\pi,\pi_0)$ is valid for $\poly(G,p)$.
    Consider a set $B:=\{x^i\}_{i\in [\ell]} \subseteq F\cap \{0,1\}^{n+ \bar e-|S|}$ of $\ell$ affine independent vectors where \mbox{$\ell = \dim(F) +1$}.
We define a vector $\bar x \in \R^{n+\bar e-|S|}$ such that $\bar x = \left(\sum_{i \in [\ell]} x^i \right)/\ell$.
    Because~$\bar x$ is a convex combination of the vectors in $B$, we have $\bar x \in F$, and hence~$(T^*+1, \bar x ) \in F'$, where $T^* = \sum_{v \in V}p(v)$.
    Moreover, it is clear that $(T^*,x^i)$ belongs to $F'$ for every~$i \in [\ell]$.
    Since the vectors in $B$ are affine independent, $\{(T^*+1, \bar x )\}\cup\{(T^*,x^i)\}_{i\in [\ell]}$ is a set of $\ell+1$ affine independent vectors in $F'$.
    Consequently, we have $\dim(F') \geq \ell+1 -1 = \dim(F)+1$.
\end{proof}

We next show how to take advantage of the polyhedral studies of $\polyR(G)$ described in the literature to investigate the facial structure of $\poly(G,p)$.

\begin{theorem}\label{thm:dim}
   $\poly(G,p)$ is full-dimensional, i.e., $\dim(\poly(G,p)) = n+\bar e -|S| +1$.
\end{theorem}
\begin{proof}
    Consider the trivial valid inequality $(\vec 0, 0) \in \R^{n+\bar e-|S|}\times \R$ of $\polyR(G)$.
    Apply Lemma~\ref{lem:lifting} with $(\pi, \pi_0) = (\vec 0, 0)$ and note that $F=\polyR(G)$ and $F'=\poly(G,p)$.
    Thus we have $\dim(\poly(G,p)) \geq \dim(\polyR(G))+1$.
    Since $\dim(\polyR(G))=n+\bar e -|S|$ as shown by~\cite{CamCamCor08}, it follows that $\dim(\poly(G,p)) = n+\bar e -|S| +1$. 
\end{proof}

\begin{theorem}\label{thm:trivial-facets}
    The following inequalities induce facets of $\poly(G,p)$:
    \begin{enumerate}[(i)]
        \item $x_{vv} \le 1$ for every $v \in V\setminus S$,
        \item $x_{vv}\ge 0$ for every $v \in T$,
        \item $x_{vu} \ge 0$ for every $v \in V\setminus T$, and $u \in \bar N^+(v)$ such that $|\bar N^-(u)|\ge 2$,
        \item $\sum_{u \in \bar N^-[v]} x_{uv} \ge 1$ for every $v \in V\setminus S$.
    \end{enumerate}
\end{theorem}
\begin{proof}
    The previous inequalities were shown to be facet-defining for $\polyR(G)$~\citep{CamCamCor08}.
    Thus the result follows from using Lemma~\ref{lem:lifting} for each of the inequalities, and from Theorem~\ref{thm:dim}, which shows that the polytope is full-dimensional.
\end{proof}

\subsection{Stable set inequalities}\label{sec:stable-set}

\cite{CamMouSan16} presented a detailed study of the facial structure of the representatives formulation 
for the $t$-fold coloring problem, a generalization of the classical vertex coloring problem where each vertex is assigned~$t\in\Z_>$ different colors.
In particular, they showed that $\polyR(G)$ inherits facet-defining inequalities from the stable set polytope.

A \emph{stable set} of a graph~$H$ is a subset of pairwise non-adjacent vertices. 
The \emph{stable set polytope} of~$H$ is defined as
\[
\stab(H)=\conv\{z\in \{0,1\}^{|V(H)|} \colon  z_u + z_v \leq 1, \; \text{ for all } \{u,v\}\in E(H)\}.
\]
It is known that the stable set polytope is full-dimensional and there are many results about its facial structure~\cite[see][]{CheVri02B,Tro75,GilTro81,GalSas97,CheVri02A,EOSV08,OriSta08}. 
Let us denote by~$G^+(u)$ the subgraph of $G$ induced by the vertices in $\bar N^+(u)$ for each~$u\in V(G)\setminus T$.

\begin{theorem}[\cite{CamMouSan16}]\label{thm:stab}
    Let $u\in V$. 
    If \(\sum_{v\in \bar N^+(u)} \lambda_v z_v \leq  \lambda_0\) is
    facet-defining for~$\stab(G^+(u))$, then 
    \(\sum_{v\in \bar N^+(u)} \lambda_v x_{uv} \leq  \lambda_0 \gamma_u\)
    is facet-defining for~$\polyR(G)$,  where $\gamma_u=1$ if~$u\in S$, and~$\gamma_u=x_{uu}$ otherwise. 
\end{theorem}

We next prove an analogous result for $\poly(G,p)$.
Essentially, we show that the facet-defining inequalities of the stable set polytope are also facet-defining for $\poly(G,p)$.

\begin{theorem}\label{thm:stab-facets}
    Let $u\in V$. 
    If \(\sum_{v\in \bar N^+(u)} \lambda_v z_v \leq  \lambda_0\) is
    facet-defining for~$\stab(G^+(u))$, then 
    \(\sum_{v\in \bar N^+(u)} \lambda_v x_{uv} \leq  \lambda_0 \gamma_u\)
    is facet-defining for~$\poly(G,p)$, where $\gamma_u=1$ if~$u\in S$, and~$\gamma_u=x_{uu}$ otherwise. 
\end{theorem}
\begin{proof}
    The claimed result follows from applying Lemma~\ref{lem:lifting} to the inequalities of of $\mathcal{R}(G)$ obtained via Theorem~\ref{thm:stab}, and using that $\poly(G,p)$ is full-dimensional as shown in Theorem~\ref{thm:dim}.
\end{proof}

The previous result guarantees that the following so-called \emph{external inequalities} in~\cite{CamCamCor08}
are valid for $\poly(G,p)$ and induce facets under certain conditions.
We remark that the external inequalities strengthen the rank inequalities introduced by~\cite{chvatal75} for the stable set polytope.
These are precisely the inequalities $\sum_{v \in U} z_v \le \alpha(G[U])$ 
 for $U \subseteq V$, where $\alpha(G[U])$ is the size of the largest stable set in the subgraph of $G$ induced by the vertices belonging to $U$.

\begin{corollary}
    Let $v \in V\setminus T$ and $U\subseteq \bar N^+(v)$.
    The following external inequality is valid for $\poly(G,p)$:
    \begin{equation}\label{ineq:external}
        \sum_{u \in U}\frac{1}{\alpha_u}x_{vu} \leq \gamma_v,
    \end{equation}
    where $\alpha_u$ denotes the maximum size of a stable set containing $u$ in $G[U]$, $\gamma_v=1$ if~$v\in S$, and~$\gamma_v=x_{vv}$ otherwise. Moreover, it is facet-defining if $U$ induces a maximal clique.
\end{corollary}

An interesting particular case of the external inequalities occurs when $G[U]$ is a complete graph (i.e., a \emph{clique}):
\begin{align}
     & \sum_{u \in U}x_{vu} \leq \gamma_v & \forall U\subseteq V \text{ such that } G[U] \text{ is a clique.} \label{ineq:clique}
\end{align}
Considering the stable set polytope, these inequalities correspond to the \emph{clique inequalities} and are facet-defining if and only if $U$ induces a maximal clique
as shown by~\cite{padberg1973}.

Another subclass of valid inequalities for $\poly(G,p)$ is obtained from~\eqref{ineq:external} when considering induced cycles in $G$ as follows. 
\begin{align}
& \sum_{u \in U}x_{vu} \leq \frac{|U|-1}{2} \gamma_v  & \forall v\in V, U\subseteq V,  \text{ such that $G[U]$ is an odd cycle.} \label{ineq:odd-cycle}
\end{align}
Indeed, these are rank inequalities as $\alpha(G[U]) = (|U|-1)/2$ if $G[U]$ is an odd cycle, and  they correspond to the \emph{odd-cycle inequalities} proposed by~\cite{padberg1973} in his polyhedral study of the stable set problem.
He mentioned that such inequalities generally do not induce facets and described a lifting procedure to strengthen them.
Moreover, Padberg showed that an odd-cycle inequality induces a facet of $\stab(G)$ if $G$ itself is an odd cycle of length at least~$5$.

\cite{LovSch90} noticed that given an inequality $(\pi,\pi_0) \in \R^n\times \R$, it is unlikely that there exists a polynomial-time algorithm to verify that $(\pi,\pi_0)$ is a rank inequality since such procedure involves computing $\alpha$, an $\NP$-hard problem in general.
The external inequalities present a similar characteristic as calculating $\alpha_v$ is not easier than calculating~$\alpha$.

We shall focus on clique and odd-cycle inequalities since a maximum stable set (containing any given vertex) of the corresponding graphs is easy to compute, and these inequalities proved to be useful in the branch-and-cut approach to \vcp\ by~\cite{MENDEZDIAZ2006826}.
The separation problem of clique inequalities is $\NP$-hard as it corresponds to solving the maximum-weighted clique problem.
On the positive side, the odd-cycle inequalities~\eqref{ineq:odd-cycle} can be separated in polynomial time following essentially the same strategy used for separating the odd-cycle inequalities of the stable set polytope as shown by~\cite{GerSch06} (see also~\cite{RebReiPar12}).
The separation algorithm for the odd-cycle inequalities~\eqref{ineq:odd-cycle} is described next.

Consider a vector $\bar x \in \R^{n+\bar e -|S|}$ satisfying~\eqref{eq:edge2:sources}, a vertex $v \in V\setminus T$, and the subgraph $G^+(v)$ induced by its positive anti-neighborhood in $G$.
For each edge $e=\{u,w\} \in E(G^+(u))$, let us define $f(e) = (\bar \gamma_v - \bar x_{vu} - \bar x_{vw})/2$, where $\bar \gamma_v = \bar x_{vv}$ if $v \notin S$, and $\bar \gamma_v = 1$ otherwise.
For every  $U \subseteq \bar N^+(v)$ such that $C=G[U]$ is an odd cycle, we have
\[f(C) = \sum_{\{u,w\} \in E(C)} f(uw) =  \sum_{\{u,w\} \in E(C)} \frac{\bar \gamma_v - \bar x_{vu} - \bar x_{vw}}{2} = \frac{|U|}{2}\bar \gamma_v - \sum_{u \in U} \bar x_{vu}.\]
Observe that $\bar x$ violates~\eqref{ineq:odd-cycle} if and only if $|U| \bar \gamma_v/2 - \sum_{u \in U } \bar x_{vu} < \bar \gamma_{v}/2$. 
Hence the separation of the odd-cycle inequalities corresponding to $v$ is equivalent to finding a minimum-weight odd cycle in  the edge-weighted graph $(G^+(v), f)$, and then verifying whether its weight is smaller than $\bar \gamma_{v}/2$.
For this, we create a bipartite graph $B$ by splitting each vertex $u$ in $G^+(v)$ into vertices $u^+$ and $u^-$, and adding to $B$, for each edge $\{u,w\}$ in $G^+(v)$,  edges $\{u^+, w^-\}$ and $\{u^-, w^+\}$ of weight $f(\{u,w\})$.
Finally, a minimum-weight odd cycle in $G^+(v)$ can be found by computing a shortest path between $u^+$ and $u^-$ in the weighted graph~$B$ for each $u \in V(G^+(v))$.

For the maximum stable set problem, \cite{Rebennack11} presented an exact separation routine for odd-hole inequalities (which are then lifted to maximal cliques in case of cycles of size 3), a heuristic separation for clique inequalities that randomly selects an edge and extend it to a maximal clique, and a heuristic based on edge-projection to separate the rank inequalities.

To conclude this section, we present the corresponding external inequalities for the assignment formulation for \pmc\ as follows.
Let $\mathcal{Q}(G,p,m)$ be the polytope associated with (AF), that is,
\[\mathcal{Q}(G,p,m) = \conv \{(\tilde x, \tilde y) \in \{0,1\}^{nm} \times \R_\ge : (\tilde x, \tilde y) \text{ satisfying } \eqref{eq:assign}, \eqref{eq:ass:edge}, \text{ and } \eqref{eq:time}\}. \]

\begin{proposition}
    Let $U\subseteq V$, and $k \in [m]$.
    The following inequality is valid for $\mathcal{Q}(G,p,m)$:
    \begin{equation}\label{ineq:ass:external}
        \sum_{u \in U}\frac{1}{\alpha_u}\tilde x_{uk} \leq 1,
    \end{equation}
    where $\alpha_u$ denotes the maximum size of a stable set containing $u$ in $G[U]$.
\end{proposition}
\begin{proof}
    Consider a vector~$\tilde x \in \{0,1\}^{nm}$ in $\mathcal{Q}(G,p,m)$. 
    Note that the jobs in $U$ assigned to machine~$k$ form a stable set~$W$.
    Hence $|W| \le \alpha_*$ where $\alpha_*=\min_{u \in W} \alpha_u$.
Then we have 
    \[\sum_{u \in U} \frac{1}{\alpha_u} \tilde x_{uk} = \sum_{u \in W} \frac{1}{\alpha_u} \leq \frac{1}{\alpha_*}  \sum_{u \in W} 1 \leq \frac{1}{\alpha_*} |W| \leq 1.\]
\end{proof}

As in the case of the representatives model, inequalities~\eqref{ineq:ass:external}  strengthen the rank inequalities, which contain the clique and odd-cycle inequalities.
The previous discussion on separation problems also holds for the corresponding inequalities with assignment variables.

\bigskip

\subsection{Subgraph-induced inequalities}

In this section, we study inequalities related to the minimum number of machines required to schedule a given subset of jobs.
In the context of the vertex coloring problem, these inequalities first appeared in~\cite{CamCamCor08}.

For every set $U \subseteq V$, we define $S_U = \{v\in U : \bar N^-(v)\cap U=\emptyset\}$.
In words, $S_U$ is the set of minimal vertices in the suborder $\prec$ restricted to $U$.

\begin{proposition}\label{prop:valid:internal}
    Let $U \subseteq V$. The following inequality is valid for $\poly(G,p)$:
    \begin{equation}\label{ineq:internal}
        \sum_{v \in U\setminus S_U} \: \sum_{ u \in (\bar N^-(v)\setminus U)\cup \{v\} } x_{uv} \geq \chi(G[U])-|S_U|.
    \end{equation}
\end{proposition}
\begin{proof}
    First note that the subgraph $G[U\setminus S_U]$ uses at least $\chi(G[U])-|S_U|$ machines as  $S_U$ induces a complete graph in $G$.
    Moreover, each of such machines is either represented by a job in $U\setminus S_U$ or by a job not in $U$.
    Therefore, inequality~\eqref{ineq:internal} is valid for $\poly(G,p)$.
\end{proof}

In what follows, we focus on the cases where $G[U]$ plays an important role in the  coloring problem and $\chi(G[U])$  is easy to compute.
Let $\ell, q \in \Z$ such that $\ell\ge 2$ and $q\ge 2\ell$.
\cite{Tro75} defined the \emph{web}~$\web$ as the graph with vertex set~$\{ v_{0}, v_{1}, \ldots, v_{q-1} \}$ and edge 
set~$\{ v_{i}v_{j} : \ell \leq |i -j| \leq q-\ell \}$.
The \emph{antiweb}~$\antiweb$ is the complement of~$\web$.
Examples of a web and antiweb are depicted in Figure~\ref{fig:web}.
Trotter also showed that~$\alpha(\web)=\ell$ and~$\alpha(\antiweb)=\left\lfloor \frac{q}{\ell} \right\rfloor$.
\cite{CamCorMouSan13} proved that $\chi(H)=\left\lceil \frac{n}{\alpha(G)} \right\rceil$ if $H=\web$ or~$H=\antiweb$.

\begin{figure}[tbh!]
\centering
\begin{subfigure}{.47\textwidth}
    \centering

\begin{tikzpicture}[
auto,
            node distance = 1cm, scale=0.8
]
        \tikzstyle{every state}=[
            draw = black,
fill = white,
            minimum size = 2mm,
            scale=0.6
        ]
    
\foreach \a in {0,1,...,9}{
\node[state] (v\a) at (-\a*360/10+90: 2cm) {${\a}$};
}

\foreach \a in {0,1,...,5}{

    \foreach \b [evaluate=\b as \i using int(\a+\b)] in {4,5,6}{\pgfmathparse{\i>9?0:\i}
        \ifnum\pgfmathresult>0
            \path (v\a) edge node{} (v\pgfmathresult);
        \fi
    }
}
\end{tikzpicture}
\caption{$W^{10}_4$}

\end{subfigure}\hfill
\begin{subfigure}{.47\textwidth}
    \centering

    \begin{tikzpicture}[
auto,
            node distance = 1cm, scale=0.8
]
        \tikzstyle{every state}=[
            draw = black,
fill = white,
            minimum size = 2mm,
            scale=0.6
        ]
    
        \foreach \a in {0,1,...,9}{
\node[state] (v\a) at (-\a*360/10+90: 2cm) {${\a}$};
}
        
        \foreach \a in {0,1,...,9}{

            \foreach \b [evaluate=\b as \i using int(\a+\b)] in {1,2,3,7,8,9}{\pgfmathparse{\i>9?0:\i}
                \ifnum\pgfmathresult>0
                    \path (v\a) edge node{} (v\pgfmathresult);
                \fi
            }
        }
    \end{tikzpicture}
    \caption{$\bar W^{10}_4$}
\end{subfigure}
\caption{A web graph and its complement.\label{fig:web}}
\end{figure}

A graph $H$ is said to be \emph{perfect} if, for every induced subgraph $H'$ of $H$, it holds that $\chi(H') = \omega(H')$, where $\omega(H')$ denotes the size of the largest clique in~$H'$.
By the strong perfect graph theorem due to~\cite{ChuRobSeyTho06}, perfect graphs can be characterized in terms of odd cycles and their complements, precisely: a graph $H$ is perfect if and only if neither $H$ nor its complement $\bar H$ contains an odd cycle of length at least $5$ as an induced subgraph.
\cite{GROTSCHEL1984325} proved that the vertex coloring problem can be solved in polynomial time on perfect graphs.

As webs and antiwebs contain cycles of size at least 4 and their complements, they represent a main obstacle when solving the coloring problem.
Similarly, one may expect that the same classes of graphs are also relevant when solving \pmc.
Indeed, we shall prove that inequalities~\eqref{ineq:internal} when $G[U]$ is a web or an antiweb induce facets of $\poly(G,p)$ under certain criticality conditions.

A graph $H$ is said to be \emph{$\chi$-critical}  if $\chi(G-v) < \chi(G)$ for all $v \in V(G)$.
\cite{CamMouSan16} showed that a graph $H$
is $\chi$-critical web~$\web$ or antiweb~$\antiweb$ if and only if  $\frac{q-1}{\alpha(H)} \in \Z$.
It is known that $\alpha(\web)=\ell$ and $\alpha(\antiweb)=\left \lfloor \frac{q}{\ell} \right \rfloor$~\cite[see][]{Tro75}.
The following facets of $\polyR(G)$ are based on induced $\chi$-critical webs and antiwebs.

\begin{theorem}[\cite{CamMouSan16}]\label{thm:webs-antiwes}
    Let $U\subseteq V(G)$.
    If~$G[U]$ is $\chi$-critical web or antiweb,
    then inequality~\eqref{ineq:internal} is facet-defining for~$\polyR(G)$.
\end{theorem}

\begin{theorem}
    Let $U\subseteq V(G)$ and  $S_U$ be the set of sources in the suborder of $\prec$ restricted to $U$, 
    that is, $S_U=\{v\in H: \bar N^-(v)\cap H=\emptyset\}$. 
If~$G[U]$ is $\chi$-critical web or antiweb,
then inequality~\eqref{ineq:internal} is facet-defining for~$\poly(G,p)$.
\end{theorem}
\begin{proof}
    The validity of~\eqref{ineq:internal} for $\poly(G,p)$ is shown in Proposition~\ref{prop:valid:internal}.
    By Theorems~\ref{thm:webs-antiwes} and Lemma~\ref{lem:lifting},  inequality~\eqref{ineq:internal} is facet-defining for~$\poly(G,p)$ since this polytope is full-dimensional (Theorem~\ref{thm:dim}).
\end{proof}

To the best of our knowledge, the existence of polynomial-time  separation algorithms for inequalities~\eqref{ineq:internal} is an open question.

\section{Computational experiments}\label{sec:experiments} 

In this section, we report on our tests with the formulations for \pmc\, and compare their performance with the other solving methods in the literature proposed by~\cite{BIANCHESSI2021105464} and by~\cite{kowalczyk2017exact}.

Our first two solution approaches are based on the formulations (AF) and (RF), which are solved by the standard branch-and-cut algorithm of a MILP solver (with the automatic cuts implemented in the solver). 
The resulting methods are denoted by $\textsc{bc-a}$ and $\textsc{bc-r}$, respectively. 
The vertex ordering used in the representatives formulation is obtained by computing an inclusion-wise maximal clique~$K$ in the input graph using the heuristic due to~\cite{Grosso08}, and then ordering the vertices in non-decreasing distances from $K$.
The motivation to use such an ordering is to reduce the number of variables in the model.

All the proposed algorithms are implemented in C++ using LEMON Graph Lib.\ 1.3.1, and Gurobi 10 as the MILP solver.
All Gurobi parameters are set as default, except for \texttt{MIPFocus}, which is set to $1$.
The experiments are carried out on a computer running Linux Ubuntu 22.04 LTS (64-bit) equipped with Intel Core i7-4790 at 3.6 GHz and 8GB of RAM\@.
The time limit for each instance is 840 seconds.
We also run the binary search algorithm devised by~\cite{kowalczyk2017exact} (denoted by $\textsc{kl}$) in the same computational environment, with the same time limit of 840 seconds.

The computational results of the branch-and-price algorithm of~\cite{BIANCHESSI2021105464} (denoted by \textsc{bt}) are obtained from the table mentioned in their work, and which is publicly available at \url{https://github.com/Treema81/PMC_Data}.
The time limit of 840 seconds is chosen for $\textsc{kl}$ and the branch-and-cut procedures so as to provide a fair comparison with the results for a limit of 720 seconds used in~\cite{BIANCHESSI2021105464} for \textsc{bt}, considering the difference between the CPUs as reported in~\cite{PassMark}.

\medskip
\subsection{Dataset}\label{sec:dataset}

We run tests on two different datasets.  
The first dataset contains the instances used in both~\cite{kowalczyk2017exact} and~\cite{BIANCHESSI2021105464}, and will be the basis for all the experiments in Sections~\ref{sec:comparison:compact} and~\ref{sec:comparison:state}.
These instances are constructed  with the Erd\H{o}s–Rényi random graph model, which is described next.  
Given $n\in \Z_>$ and $d \in [0,1]$, $\mathcal{G}(n,d)$ is a random variable whose values are $n$-vertex graphs.
The probability distribution of $\mathcal{G}(n,d)$ is defined as follows: for any graph $G$ on $n$ vertices with $\ell$ edges, the probability that $\mathcal{G}(n,d) = G$ is $d^{\ell}(1-d)^{\binom{n}{2}-\ell}$.
The processing times of the jobs are integer values chosen uniformly at random from the intervals $a=[1,10]$, $b = [1,50]$, or $c=[1,100]$.
The dataset used in the literature~\citep{kowalczyk2017exact,BIANCHESSI2021105464} contains 2550 instances, and is constructed as follows: it contains~10 instances for each pair $(n,d) \in \{10,25,50,75,100\} \times \{0.1,0.2,0.3,0.4,0.5\}$, processing-time intervals in $\{a,b,c\}$, and $m \in \{5,10,15,20\}$ if $n\ge 25$,  and $m=5$ otherwise.

We observe that most of these instances have their optimal value very close to (at most 5\% larger than) the trivial lower bound~$\lceil\sum_{v \in V} p(v)/m\rceil$ for \pmc. This means that the conflicts in these instances do not affect the optimal value much compared to the optimal value when scheduling the same jobs without conflicts.
For this reason, we consider in Section~\ref{subsec:resultsstructeredinstances} a second set of instances where the optimal value is much larger than the trivial bound.

\subsection{Comparison of the compact models}\label{sec:comparison:compact}
\textcolor{black}{For the computational experiments reported in this subsection, we set the Gurobi parameter \texttt{Threads} equal to six.}
Out of the 2550 benchmark instances used in the literature, \textsc{bc-r} solves 100 more instances than \textsc{bc-a}, and produces gaps that are (on average) 13\% smaller than those obtained by~\textsc{bc-a}.
The average runtimes per instance of \textsc{bc-a} and \textsc{bc-r} are 205.08 and 178.32 seconds, respectively.
Both algorithms solve all instances with five machines; therefore, the poor performance of \textsc{bc-a} compared to \mbox{\textsc{bc-r}} is due to instances with at least 10 machines.
This may be explained by the symmetries in the assignment formulation, which increase with the number of machines.
Preliminary experiments indicate that the symmetry-breaking inequalities in Section 2 have minimal effect on the performance of~\textsc{bc-a}.

Since testing the feasibility of \pmc\ instances is already $\NP$-complete, we run the primal heuristics proposed by~\cite{kowalczyk2017exact} before starting our branch-and-cut procedure.
These heuristics include feasibility tests of the \pmc\ instances using algorithms for vertex coloring (e.g. DSATUR), the approximation algorithm for \pmc\ devised by~\cite{BODLAENDER1994219}, and a local search heuristic that swaps groups of jobs between two machines.
If any solution is found within five seconds, we use it as a ``warm start'' in the proposed solving methods.
We have also implemented separation algorithms for the clique inequalities~\eqref{ineq:clique} and odd cycle inequalities~\eqref{ineq:odd-cycle}.
The clique inequalities are separated exactly using integer linear programming, and the odd-cycle inequalities are separated in polynomial time as described in Section~\ref{sec:stable-set}.
After preliminary experiments, we notice that the odd-cycle inequalities are not very effective for reducing the gap in most instances of our dataset. 
In general, the best results are obtained when we only separate clique inequalities at the root node of the enumeration tree.
As mentioned in Section~\ref{sec:dataset}, many of these instances have their optimal value close to the trivial lower bound for \pmc, which also equals the dual bounds obtained by the proposed algorithms already at the root node.
Therefore, we only separate them on instances with density at least $0.3$.

Let us denote by \textsc{bc-a+} and \textsc{bc-r+} the algorithms that include the primal heuristics and the clique cuts as described above.
The inclusion of the heuristics and cuts has more impact on the representatives formulation:  \textsc{bc-a+} solves five more instances than \textsc{bc-a}, while \textsc{bc-r+} solves 68 more instances than  \mbox{\textsc{bc-r}}.
Overall, \textsc{bc-r+} solves 163 more benchmark instances than \textsc{bc-a+}, and the average runtimes of \textsc{bc-r+} and \textsc{bc-a+} are 152.06 and 204.98 seconds, respectively.
These results show that  \textsc{bc-r+} clearly dominates the other tested methods based on the compact models.
Next, we compare the performance of this algorithm with the state-of-the-art solving methods for \pmc.

\subsection{Comparison with the state of the art}\label{sec:comparison:state}

\textcolor{black}{Since \textsc{kl} uses a single thread, we set the Gurobi parameter \texttt{Threads} equal to one for~\textsc{bc-r+} and  \textsc{bc-a+} in the experiments comparing their performance with \textsc{kl} and~\textsc{bt} in Sections~\ref{sec:comparison:state} and~\ref{subsec:resultsstructeredinstances}. 
We observe that \textsc{bt} uses six threads, as stated in~\cite{BIANCHESSI2021105464}. }
Table~\ref{table:overall:new} shows the number of instances for which each algorithm reaches the time limit  without finding an optimal solution or proving infeasibility (column \emph{unsolved}).
For the unsolved instances where the algorithm computes a feasible solution (\emph{subopt}),  the average gaps obtained are presented in the column labeled \emph{Gap}.  We report the results in this way because if we take the average including all instances with gap~$0$ then the three methods are indistinguishable.
The gaps are reported for each value $d \in \{0.1,0.2,0.3,0.4,0.5\}$ separately.   
The gap for a method on an instance is the difference (expressed as a percentage) between the best primal and dual bounds found by this particular method within the runtime limit.
We remark that the best dual bounds produced by the three methods are essentially the same for all benchmark instances.

The average gaps obtained by \textsc{bc-r}+, \textsc{kl}, and  \textsc{bt} are very similar when we consider the entire benchmark of 2550 instances:  close to  $0\%$  if $d \in \{0.1,0.2\}$, around $0.5\%$ if $d=0.3$, and approximately $1\%$ if $d \in \{0.4,0.5\}$.
These small average gaps are due to the fact that more than $80\%$ of the instances in the benchmark are solved to optimality by any of these three algorithms.
We observe that, for every benchmark instance, both \textsc{bc-r+} and \textsc{kl} either show infeasibility or give a primal bound.
In contrast, \textsc{bt} neither proves infeasibility nor obtains a primal solution for 84 instances (30 with $d=0.1$, 23 with $d=0.3$, 10 with $d=0.4$, and 21 with $d=0.5$).

As indicated in Table~\ref{table:overall:new}, the average gaps obtained by \textsc{bc-r}+ are smaller than the ones obtained by \textsc{kl} (except when the density equals $0.2$) and \textsc{bt}.
Indeed, for the instances with density at least $0.3$, the average gaps obtained by the proposed algorithm are up to \textcolor{black}{8} times smaller than the ones produced by \textsc{kl} and up to \textcolor{black}{three} times smaller than the gaps given by \textsc{bt}. 
The box plots in Figure~\ref{fig:boxplot} show that most of the solutions produced by \textsc{bc-r+} on these instances have a gap of at most 5\%.
Moreover, it indicates that \textsc{bc-r+} produces better solutions compared to those given by \textsc{kl} and \textsc{bt}. 
On the other hand, the total number of instances not solved to optimality by \textsc{bc-r}+ is around three times larger than that obtained by \textsc{kl} and twice the number given by~\textsc{bt}. 

\begin{table}[tbh!]
\small
\setlength{\tabcolsep}{2pt}
\caption{For each of the algorithms \textsc{bc-r}+, \textsc{kl}, and \textsc{bt}, the table shows the average gaps (in \%) of the benchmark instances not solved to optimality by that method. For every graph density in $\{0.1,0.2,0.3, 0.4, 0.5\}$, there are 510 instances in the literature dataset. }
\rowcolors{1}{lightgray}{}
\centering
\begin{tabular}{lrrrrrr}
\hiderowcolors
\toprule
 & \multicolumn{2}{c}{\textsc{bc-r}+} & \multicolumn{2}{c}{\textsc{kl}} & \multicolumn{2}{c}{\textsc{bt}}\\ 
\cmidrule(r){2-3} \cmidrule(r){4-5} \cmidrule(r){6-7} 
$d$  & Unsolved & Gap \%   & Unsolved & Gap \%  & Unsolved  &  Gap \% (subopt)\\
\midrule
0.1   & \textcolor{black}{9} &  \textcolor{black}{0.55}  & 21 & 2.77  & 36  &  0.44  \phantom{0}(6)\\
0.2   & \textcolor{black}{57} &  \textcolor{black}{0.41}  & 0 & -  &  28  &  1.07 (28)\\
0.3   & \textcolor{black}{120} &  \textcolor{black}{1.66}  & 35 & 14.10  & 60  &  6.07 (37)\\
0.4   & \textcolor{black}{137} &  \textcolor{black}{3.30}  & 41 & 10.58  & 61  &  12.34 (51)\\
0.5  & \textcolor{black}{138} &  \textcolor{black}{5.23}  & 38 & 9.84  & 69 &  5.99 (48)  \\
\noalign{\global\rownum=1}
\showrowcolors
\bottomrule
\end{tabular}
\label{table:overall:new}
\end{table}

\begin{figure*}[tbh!]
    \centering

\begin{subfigure}[t]{0.2\textwidth}
\centering
\begin{tikzpicture}
	\pgfplotstableread[col sep=comma]{data-bcr3.csv}\csvdataA
\pgfplotstabletranspose\datatransposedA{\csvdataA} 

	\pgfplotstableread[col sep=comma]{data-kl3.csv}\csvdataB
\pgfplotstabletranspose\datatransposedB{\csvdataB} 

    \pgfplotstableread[col sep=comma]{data-bt3.csv}\csvdataC
\pgfplotstabletranspose\datatransposedC{\csvdataC}
    
	\pgfplotstableread[col sep=comma]{data-bcr3ST.csv}\csvdataD
\pgfplotstabletranspose\datatransposedD{\csvdataD}

	\begin{axis}[
        y=0.1cm,
        x=0.5cm,
        xminorgrids=true,
		boxplot/draw direction = y,
		x axis line style = {opacity=0},
		axis x line* = bottom,
		axis y line* = left,
enlarge y limits={abs=0.4cm},
		ymajorgrids,
		xtick = {1, 2, 3},
		xticklabel style = {align=center, font=\footnotesize, rotate=45},
		xticklabels = {\textsc{bc-r+}, \textsc{kl}, \textsc{bt}},
		xtick style = {draw=none}, ylabel = {Gap \%},
        yticklabel style = {scale=0.7},
        y label style={scale=0.7},
        boxplot/whisker range=40, ]
\addplot+[boxplot, fill, draw=black] table[y index=1] {\datatransposedD};
        \addplot+[boxplot, fill, draw=black] table[y index=1] {\datatransposedB};
        \addplot+[boxplot, fill, draw=black] table[y index=1] {\datatransposedC};
	\end{axis}
\end{tikzpicture}
\caption{$d=0.3$}
\end{subfigure}  \qquad

\begin{subfigure}[t]{0.2\textwidth}
\centering
\begin{tikzpicture}
	\pgfplotstableread[col sep=comma]{data-bcr4.csv}\csvdataA
\pgfplotstabletranspose\datatransposedA{\csvdataA} 

	\pgfplotstableread[col sep=comma]{data-kl4.csv}\csvdataB
\pgfplotstabletranspose\datatransposedB{\csvdataB} 

    \pgfplotstableread[col sep=comma]{data-bt4.csv}\csvdataC
\pgfplotstabletranspose\datatransposedC{\csvdataC}
    
	\pgfplotstableread[col sep=comma]{data-bcr4ST.csv}\csvdataD
\pgfplotstabletranspose\datatransposedD{\csvdataD}

	\begin{axis}[
        y=0.1cm,
        x=0.5cm,
        xminorgrids=true,
		boxplot/draw direction = y,
		x axis line style = {opacity=0},
		axis x line* = bottom,
		axis y line* = left,
enlarge y limits={abs=0.4cm},
		ymajorgrids,
		xtick = {1, 2, 3},
		xticklabel style = {align=center, font=\footnotesize, rotate=45},
		xticklabels = {\textsc{bc-r+}, \textsc{kl}, \textsc{bt}},
		xtick style = {draw=none}, ylabel = {Gap \%},
        yticklabel style = {scale=0.7},
        y label style={scale=0.7},
        boxplot/whisker range=20, ]
\addplot+[boxplot, fill, draw=black] table[y index=1] {\datatransposedD};
        \addplot+[boxplot, fill, draw=black] table[y index=1] {\datatransposedB};
        \addplot+[boxplot, fill, draw=black] table[y index=1] {\datatransposedC};
	\end{axis}
\end{tikzpicture}
\caption{$d=0.4$}
\end{subfigure}  \qquad

\begin{subfigure}[t]{0.2\textwidth}
\centering
\begin{tikzpicture}
	\pgfplotstableread[col sep=comma]{data-bcr5.csv}\csvdataA
\pgfplotstabletranspose\datatransposedA{\csvdataA} 

	\pgfplotstableread[col sep=comma]{data-kl5.csv}\csvdataB
\pgfplotstabletranspose\datatransposedB{\csvdataB} 

    \pgfplotstableread[col sep=comma]{data-bt5.csv}\csvdataC
\pgfplotstabletranspose\datatransposedC{\csvdataC}

	\pgfplotstableread[col sep=comma]{data-bcr5ST.csv}\csvdataD
\pgfplotstabletranspose\datatransposedD{\csvdataD}
    
	\begin{axis}[
        y=0.1cm,
        x=0.5cm,
        xminorgrids=true,
		boxplot/draw direction = y,
		x axis line style = {opacity=0},
		axis x line* = bottom,
		axis y line* = left,
enlarge y limits={abs=0.4cm},
		ymajorgrids,
		xtick = {1, 2, 3},
		xticklabel style = {align=center, font=\footnotesize, rotate=45},
		xticklabels = {\textsc{bc-r+}, \textsc{kl}, \textsc{bt}$\bigstar$},
		xtick style = {draw=none}, ylabel = {Gap \%},
        yticklabel style = {scale=0.7},
        y label style={scale=0.7},
        boxplot/whisker range=20, ]
\addplot+[boxplot, fill, draw=black] table[y index=1] {\datatransposedD};
        \addplot+[boxplot, fill, draw=black] table[y index=1] {\datatransposedB};
        \addplot+[boxplot, fill, draw=black] table[y index=1] {\datatransposedC};
	\end{axis}
\end{tikzpicture}
\caption{$d=0.5$}
\end{subfigure} \caption{Gaps produced by \textsc{bc-r+}, \textsc{kl}, and \textsc{bt} on dense instances ($d\ge 0.3$) where the corresponding algorithm computes a primal solution. 
    $\bigstar$ Recall that \textsc{bt} is the only method that does not find a feasible solution to all these instances (precisely, 23 instances with $d=0.3$, 10 with $d=0.4$, and 21 with $d=0.5$.)}
    \label{fig:boxplot}
\end{figure*}

Both~\cite{kowalczyk2017exact} and~\cite{BIANCHESSI2021105464}  observe that the hardest instances for their algorithms are the ones where the number of machines $m$ is close to the chromatic number of the conflict graph. 
In what follows, we take a closer look at these instances, which are defined by the following random graphs and number of machines: 
\begin{enumerate*}[(i)] \item $\mathcal{G}(100,0.1)$ with $m=5$,
\item $\mathcal{G}(50,0.2)$ with $m=5$, 
\item $\mathcal{G}(100,0.3)$ with $m=10$, 
\item $\mathcal{G}(75,0.4)$ with $m=10$, and
\item $\mathcal{G}(100,0.5)$ with $m=15$.
\end{enumerate*}
The processing times of the jobs are integer values chosen uniformly at random from the intervals $a=[1,10]$, $b = [1,50]$, or $c=[1,100]$.
As empirically shown by Kowalczyk and Leus (see Table~5 in~\cite{kowalczyk2017exact}),  and by Bianchessi and Tresoldi (see Table~4 in \cite{BIANCHESSI2021105464}), for each $d \in \{0.1,0.2,0.3,0.4,0.5\}$, these instances are the ones where both methods for \pmc\ have the highest average gap or running time among all instances of density~$d$.
There are 150 such instances in total: 10 instances for each pair $(n,d) \in \{(100,0.1), (50,0.2), (100,0.3), (75,0.4), (100,0.5)\}$ and each of the three possible intervals of processing times $a,b,$ and $c$.
The conflict graphs in the instances are denoted by rand\_$n$\_$d$\_$i$\_$j$, where $i \in \{a,b,c\}$ and $j \in \{0,\ldots, 9\}$.

Since the sparse instances (with $d\le 0.2$) are all solved to optimality by \textsc{bc-r+}, we report the runtimes obtained by this algorithm,  \textsc{kl}, and \textsc{bt} (see Section~\ref{subsec:hardinstances_sparse}).
Most of the dense instances ($d\ge 0.3$) are not solved to optimality by any of the three methods, and so we compare the average gaps in Section~\ref{subsec:hardinstances_dense}.
All average gaps reported in the following are taken over all instances where the corresponding algorithm finds a feasible solution, including those solved to guaranteed optimality.

\subsubsection{Hard instances with sparse conflict graphs ($d\le 0.2$)} \label{subsec:hardinstances_sparse}

We first consider the largest instances (100 vertices) in the dataset with  $d =0.1$ and $m =5$.
The running times (in seconds) of \textsc{bc-r+}, \textsc{kl}, and \textsc{bt} on these instances are depicted in Figure~\ref{fig:100-5}.
Observe that \textsc{bc-r+} solves all these instances \textcolor{black}{but one} to guaranteed optimality, \textsc{kl} finds optimal solutions to nine (out of 30) instances, and \textsc{bt} does not produce any feasible solution at all within the time limit. 
Considering only the instances where \textsc{kl} computed a (guaranteed) optimal solution,  \textsc{bc-r+} is over 20 times faster than \textsc{kl}.


\begin{figure}[tbh!]
\centering
\centerline{\begin{tikzpicture}
\centering
\begin{axis}[
    y=0.6cm,
    x=0.37cm,
enlarge x limits={abs=0.3cm},
    legend style={at={(0.67,0.07)},anchor=west,legend columns=-1,nodes={scale=0.7}, fill=none, draw=none},
    ymode=log,
    xlabel=Instance,
    ylabel=Time (s),
y label style={at={(axis description cs:-0.1,.5)},anchor=north},
    symbolic x coords={a\_0,a\_1,a\_2, a\_3, a\_4, a\_5, a\_6, a\_7, a\_8, a\_9, 
    b\_0,b\_1, b\_2, b\_3, b\_4, b\_5, b\_6, b\_7, b\_8, b\_9,
    c\_0,c\_1, c\_2, c\_3, c\_4, c\_5, c\_6, c\_7, c\_8, c\_9},
xtick=data,
x label style={at={(axis description cs:.5,-0.08)},scale=0.7,anchor=north},
    y label style={at={(axis description cs:-0.08,.5)},scale=0.7,anchor=north},
    yticklabel style={scale=0.7},
xticklabel style={rotate=0,scale=0.7},
minor xtick={a\_9,b\_9},
    xminorgrids=true,
    ]
\addplot table [x=name, y=bcrST]{data100-5.dat};
\addlegendentry{\textsc{bc-r+}}
\addplot table [x=name, y=kl]{data100-5.dat};
\addlegendentry{\textsc{kl}}
\addplot table [x=name, y=bt]{data100-5.dat};
\addlegendentry{\textsc{bt}}
\end{axis}
\end{tikzpicture}}
\caption{Running times (with logarithmic scale) of  \textsc{bc-r}+,  and \textsc{kl} on random graphs with~$100$ vertices, $d=0.1$, and $m=5$.\label{fig:100-5}}
\end{figure}

%

Figure~\ref{fig:50} shows a comparison of the running times (in seconds) of \textsc{bc-r+}, \textsc{kl}, and \textsc{bt}
when solving the instances with 50 vertices, $d=0.2$, and $m=5$. 
We observe that \textsc{bt} does not find a (guaranteed) optimal solution within the time limit for nine (out of 30) instances, while \textsc{bc-r+} and \textsc{kl} produce an optimal solution for every instance with 50 vertices.
Moreover, \textsc{bc-r+}  is over three times faster than~\textsc{kl} (on average).


\begin{figure}[tbh!]
\centering
\centerline{\begin{tikzpicture}
\centering
\begin{axis}[
    y=0.35cm,
    x=0.37cm,
enlarge x limits={abs=0.3cm},
    legend style={at={(1,0.2)},anchor=east, nodes={scale=0.7}, draw=none, fill=none},
    ymode=log,
    xlabel=Instance,
    ylabel=Time (s),
symbolic x coords={a\_0,a\_1,a\_2, a\_3, a\_4, a\_5, a\_6, a\_7, a\_8, a\_9, 
    b\_0,b\_1, b\_2, b\_3, b\_4, b\_5, b\_6, b\_7, b\_8, b\_9,
    c\_0,c\_1, c\_2, c\_3, c\_4, c\_5, c\_6, c\_7, c\_8, c\_9},
xtick=data,
x label style={at={(axis description cs:.5,-0.08)},scale=0.7,anchor=north},
    y label style={at={(axis description cs:-0.08,.5)},scale=0.7,anchor=north},
    yticklabel style={scale=0.7},
xticklabel style={rotate=0,scale=0.7},
minor xtick={a\_9,b\_9},
xminorgrids=true
    ]
\addplot table [x=name, y=bcrST]{data50-new.dat};
\addlegendentry{\textsc{bc-r+}}
\addplot table [x=name, y=kl2]{data50-new.dat};
\addlegendentry{\textsc{kl}}
\addplot table [x=name, y=bt]{data50-new.dat};
\addlegendentry{\textsc{bt}}
\end{axis}
\end{tikzpicture}}
\caption{Running times (with logarithmic scale) of  \textsc{bc-r+}, \textsc{kl}, and \textsc{bt} on random graphs with~$50$ vertices, $d=0.2$, and $m=5$.\label{fig:50}}
\end{figure}

%

\subsubsection{Hard instances with dense conflict graphs ($d\ge 0.3$)} \label{subsec:hardinstances_dense}

Figure~\ref{fig:100-10} presents the average gaps (in \%) obtained by \textsc{bc-r}+, \textsc{kl}, and \textsc{bt} on the instances having processing times in the same interval ($a=[1,10]$, $b=[1,50]$, and $c=[1,100]$), 100 vertices, $d=0.3$, and $m=10$. 
For these instances, none of the algorithms finds a guaranteed optimal solution, and the gaps produced by \textsc{bc-r}+ are \textcolor{black}{around three} times smaller than the ones obtained by \textsc{kl} and \textsc{bt}.
The procedures \textsc{bc-r}+ and \textsc{kl} find feasible solutions to all tested instances, and \textsc{bt} does not find a solution to 23 (out of 30) instances. 
The clique inequalities reduce the gaps by around $20\%$ (on average).

Figure~\ref{fig:75} shows the average gaps (in \%) produced by the same algorithms on the instances with 75 vertices, $d=0.4$, and $m=10$.
We remark that all algorithms exceed the time limit on these instances (except for \textsc{kl}, which finds optimal solutions for two instances), and the dual bounds given by \textsc{kl} are slightly better (at most one unit larger) than the ones obtained by the other methods for all instances. 
Overall, the smallest gaps for these instances are obtained by \textsc{bc-r}+ and \textsc{kl}.
The former produces the best primal bounds on \textcolor{black}{12} (out of 30) instances, while the latter finds the best feasible solutions on \textcolor{black}{18} instances.
We note that \textsc{bt} does not find any solution for $10$ (out of $30$)  instances, while the other solving methods do give primal bounds for all instances.
Additionally, the gaps produced by \textsc{bt} are (on average) two times larger than the ones given by \mbox{\textsc{bc-r}+} and \textsc{kl}.
The clique inequalities are effective for these instances: \textsc{bc-r}+ produces an average gap that is $18\%$ smaller than the one obtained by~\textsc{bc-r}.

Finally, we consider the instances of 100 vertices with $d=0.5$ and $m=15$.
We observe that the clique cuts used in \textsc{bc-r}+ produce an average gap that is~$16\%$ smaller than that obtained by~\textsc{bc-r}.
As illustrated in Figure~\ref{fig:100-15-warm}, \textsc{bc-r}+ achieves smaller gaps for the instances with processing times in the interval~$a$, but \textsc{kl} still obtains tighter gaps in general for the instances with processing times in $b$ and~$c$.

\begin{figure*}[tbh!]
    \centering
\begin{subfigure}[t]{0.47\textwidth}
\centering
\begin{tikzpicture}
\begin{axis}[
bar width=.4cm,
    enlarge x limits={abs=1cm},
width=1.2\textwidth,
    height=1\textwidth,
ybar,
    legend style={at={(0.85,0.97)},anchor=north, legend columns=1, nodes={scale=0.7}, draw=none, fill=none},
nodes near coords,
    xlabel=Instance,
    ylabel=Gap \%,
symbolic x coords={$a$,$b$,$c$},
    xtick=data,
    x label style={at={(axis description cs:.5,-0.12)},scale=0.7,anchor=north},
    y label style={at={(axis description cs:-0.15,.5)},scale=0.7,anchor=north},
    yticklabel style={scale=0.7},
    xticklabel style={scale=0.7},
    every node near coord/.append style={font=\scriptsize},
    cycle list/Set1-5
]
\addplot+[fill] table [x=name, y=bcrST]{data100-10.dat};
\addlegendentry{\textsc{bc-r}+}
\addplot+[fill] table [x=name, y=kl2]{data100-10.dat};
\addlegendentry{\textsc{kl}}
\addplot+[fill] table [x=name, y=bt]{data100-10.dat};
\addlegendentry{\textsc{bt}}
\end{axis}
\end{tikzpicture}
\caption{Random graphs with~$100$ vertices, $d=0.3$, and $m=10$. \label{fig:100-10}}
\end{subfigure} \hfill
    \begin{subfigure}[t]{0.47\textwidth}

\centering
\begin{tikzpicture}
\centering
\begin{axis}[
bar width=.4cm,
    enlarge x limits={abs=1cm},
width=1.2\textwidth,
    height=1\textwidth,
ybar,
    legend style={at={(0.15,0.97)},anchor=north, legend columns=1, nodes={scale=0.7}, draw=none, fill=none},
nodes near coords,
    xlabel=Instance,
    ylabel=Gap \%,
symbolic x coords={$a$,$b$,$c$},
    xtick=data,
    x label style={at={(axis description cs:.5,-0.12)},scale=0.7,anchor=north},
    y label style={at={(axis description cs:-0.15,.5)},scale=0.7,anchor=north},
    yticklabel style={scale=0.7},
    xticklabel style={scale=0.7},
    every node near coord/.append style={font=\scriptsize},
    cycle list/Set1-5
    ]
\addplot+[fill] table [x=name, y=bcrST]{data75gap.dat};
\addlegendentry{\textsc{bc-r}+}
\addplot+[fill] table [x=name, y=kl2]{data75gap.dat};
\addlegendentry{\textsc{kl}}
\addplot+[fill] table [x=name, y=bt]{data75gap.dat};
\addlegendentry{\textsc{bt}}
\end{axis}
\end{tikzpicture}
\caption{Random graphs with~$75$ vertices, $d=0.4$, and $m=10$. \label{fig:75}}
\end{subfigure} 	\begin{subfigure}[t]{0.47\textwidth}
\centering
\begin{tikzpicture}
\centering
\begin{axis}[
bar width=.4cm,
    enlarge x limits={abs=1cm},
width=1.2\textwidth,
    height=1\textwidth,
ybar,
    legend style={at={(0.15,0.97)},anchor=north, legend columns=1,nodes={scale=0.7}, fill=none,draw=none},
nodes near coords,
    xlabel=Instance,
    ylabel=Gap \%,
symbolic x coords={$a$,$b$,$c$},
    xtick=data,
    x label style={at={(axis description cs:.5,-0.12)},scale=0.7,anchor=north},
    y label style={at={(axis description cs:-0.15,.5)},scale=0.7,anchor=north},
    yticklabel style={scale=0.7},
xticklabel style={scale=0.7},
    every node near coord/.append style={font=\scriptsize},
    cycle list/Set1-5
]
\addplot+[fill] table [x=name, y=bcrST]{data100-15.dat};
\addlegendentry{\textsc{bc-r}+}
\addplot+[fill] table [x=name, y=kl2]{data100-15.dat};
\addlegendentry{\textsc{kl}}
\addplot+[fill] table [x=name, y=bt]{data100-15.dat};
\addlegendentry{\textsc{bt}}
\end{axis}
\end{tikzpicture}
\caption{Random graphs with~$100$ vertices, $d=0.5$, and $m=15$. \label{fig:100-15-warm}}
\end{subfigure} 
\caption{Average gaps (in \%) for each processing time interval ($a=[1,10]$, $b=[1,50]$, and $c=[1,100]$) obtained by \textsc{bc-r}+, \textsc{kl}, and \textsc{bt} on dense conflict graphs.}
\end{figure*}

\subsection{Computational results on structured instances} \label{subsec:resultsstructeredinstances}

From the computational results in the previous section, we notice that most of the benchmark instances in the literature of \pmc\ with at least 50 vertices have their optimal value very close to (at most 5\% larger than) the trivial lower bound~$\lceil\sum_{v \in V} p(v)/m\rceil$ for \pmc. This means that the conflicts in these instances do not affect the optimal value much compared to the optimal value when scheduling the same jobs without conflicts.

To better understand the performance of the algorithms when the optimal value is  significantly larger than the trivial bound, we consider \pmc\ instances consisting of the following random graphs.
Given $n \in \Z_>$ an even number and $d \in [0,1]$, $\mathcal{B}(n,d)$ is a random bipartite graph with bipartition, say $\{A,B\}$, such that $|A|=|B|=n/2$, and in which every edge (linking a vertex in $A$ to one in~$B$) appears with probability~$d$.
As an example, consider an instance formed by a random bipartite graph $\mathcal{B}(100,1)$ (i.e., a complete bipartite graph where each element of the bipartition has 50 vertices), all processing times equal to one, and three machines. 
The trivial lower bound for this instance is $34$ while its optimal value is $50$.
On the other hand, if the conflict graph is $\mathcal{B}(100,0)$ (i.e., the empty graph with 100 vertices), then the lower bound equals the optimal value.

Inspired by the foregoing simple examples, we consider a set of $300$ instances defined as follows.
For each $d \in \{0.1,0.2,0.3,0.4,0.5\}$, $m \in \{3,5\}$, and $i \in \{10,50,100\}$, we have 10 instances consisting of $m$ machines and a random bipartite graph $\mathcal{B}(50,d)$ with the jobs in~$A$ and~$B$ having processing time equal to one and to $i$, respectively.
Note that the higher the probability, the larger the difference between the optimal value and the trivial lower bound.
This gap ranges from 5\% to 35\% for the instances described above.
We remark that \pmc\ restricted to bipartite conflict graphs and $m\geq 3$ cannot be approximated within $2-\varepsilon$ for all $\varepsilon >0$, unless $\Pclass=\NP$~\citep{BODLAENDER1994219}.

In the remainder of this section, we compare the computational performance of \textsc{bc-a}+, \textsc{bc-r}+, and \textsc{kl} on the instances described above.
Note that we do not report on the computational results obtained with \textsc{bt} as the implementation of~\cite{BIANCHESSI2021105464} is not available.
For the instances with $m=3$, \mbox{\textsc{bc-a}+} and \textsc{bc-r}+ solve all 150 instances within the time limit, while~\textsc{kl} does not solve 38 instances.
As shown in Table~\ref{table:bipartite:3}, \textsc{bc-r}+ and \textsc{bc-a}+ are both \textcolor{black}{very} fast on instances with $d\le 0.2$,  and \textsc{bc-r}+ is at least \textcolor{black}{15} times faster than \textsc{bc-a}+, and  \textcolor{black}{46} times faster than \textsc{kl} (on average) on instances with $d\ge 0.3$.

\begin{table}[tbh!]
\small
\setlength{\tabcolsep}{5pt}
\caption{For each of the algorithms \textsc{bc-a}+, \textsc{bc-r}+, and \textsc{kl} the table shows the average running times, and gaps (in \%) of the bipartite instances not solved to optimality by that method with $m=3$.
}
\rowcolors{1}{lightgray}{}
\centering
\begin{tabular}{lrrrr}
\hiderowcolors
\toprule
& \multicolumn{1}{c}{\textsc{bc-a}+} & \multicolumn{1}{c}{\textsc{bc-r}+} & \multicolumn{2}{c}{\textsc{kl}} \\ 
\cmidrule(r){2-2} \cmidrule(r){3-3} \cmidrule(r){4-5} 
$d$  & Time  & Time &  Time & Gap \%  \\
\midrule
0.1  & \textcolor{black}{0.51}  &  \textcolor{black}{8.53}  & 1.56 & -  \\
0.2   & \textcolor{black}{0.45}  &  \textcolor{black}{11.32}  & 599.76 & 2.19 (12) \\
0.3   & \textcolor{black}{91.91}  &  \textcolor{black}{14.56}  & 689.31 & 1.65 (21) \\
0.4   & \textcolor{black}{139.88}  &  \textcolor{black}{4.15}  & 265.84 & 1.28\phantom{0} (5)\\
0.5   & \textcolor{black}{99.88}  &  \textcolor{black}{3.33}  & 72.92 & -  \\
\noalign{\global\rownum=1}
\showrowcolors
\bottomrule
\end{tabular}
\label{table:bipartite:3}
\end{table} 
Considering the instances with $m=5$, \textsc{bc-r}+ does not solve \textcolor{black}{41} (out of 150) bipartite instances, while \textsc{bc-a}+ and \textsc{kl} do not solve \textcolor{black}{59} and 51 instances, respectively.
Moreover, the average gaps produced by \textsc{bc-r}+ are \textcolor{black}{around} 2 times smaller than the average gaps produced by \textsc{bc-a}+ and \textsc{kl}.
The average running times and gaps for each $d \in \{0.1, 0.2, 0.3, 0.4, 0.5\}$ are indicated in Table~\ref{table:bipartite:5}.

\begin{table}[tbh!]
\small
\setlength{\tabcolsep}{5pt}
\caption{For each of the algorithms \textsc{bc-a}+, \textsc{bc-r}+, and \textsc{kl} the table shows the average running times and gaps (in \%) of the bipartite instances not solved to optimality by that method with $m=5$.
}
\rowcolors{1}{lightgray}{}
\centering
\begin{tabular}{lrrrrrr}
\hiderowcolors
\toprule
& \multicolumn{2}{c}{\textsc{bc-a}+} & \multicolumn{2}{c}{\textsc{bc-r}+} & \multicolumn{2}{c}{\textsc{kl}}\\ 
\cmidrule(r){2-3} \cmidrule(r){4-5} \cmidrule(r){6-7}
$d$  & Time & Gap \% & Time & Gap \% &  Time & Gap \%  \\
\midrule
0.1  & 0.01        & -     &  \textcolor{black}{0.39}      & -         & 0.01 & -  \\
0.2  & 0.02        & -     &  \textcolor{black}{0.83}      & -         & 8.04 & -  \\
0.3  & \textcolor{black}{6.43} & -     &  \textcolor{black}{14.31}       & -         & 45.21 & 27.71 (1) \\
0.4  & 840.00  & \textcolor{black}{14.24} (30)&  \textcolor{black}{666.73}     & \textcolor{black}{4.3 (15)} & 746.31 & 7.24  (21)\\
0.5  & \textcolor{black}{839.20}  & \textcolor{black}{22.84 (29)}&  \textcolor{black}{781.62}    & \textcolor{black}{11.62 (26)} & 826.12 & 17.70 (29)  \\
\noalign{\global\rownum=1}
\showrowcolors
\bottomrule
\end{tabular}
\label{table:bipartite:5}
\end{table}

As shown in Sections~\ref{sec:comparison:compact} and~\ref{sec:comparison:state}, \textsc{bc-r+}  outperforms \textsc{bc-a+} (see Table~\ref{table:overall:new}), and produces smaller gaps than the ones  given by the state-of-the-art methods \textsc{kl} and \textsc{bt} on the benchmark instances (see Figure~\ref{fig:boxplot}).
We recall that for these instances, the gap between the optimal value  and the trivial bound is at most~5\%, and so the conflicts do not affect the the optimal value much compared to the optimal makespan when scheduling the same jobs without conflicts.
For instances where the conflicts really have an impact on the objective, that is,  for which the gap between the optimal value and the trivial bound is at least 5\%, our results indicate that \textsc{bc-r+} outperforms both \textsc{bc-a}+ and \textsc{kl}.

\section{Concluding remarks and further research} 

In this paper, we have  explored the close connection between parallel machine scheduling with conflicts (\textsc{pmc}) and the classical vertex coloring problem (\textsc{vcp}) through the lens of integer linear programming to obtain both theoretical and computational contributions to \textsc{pmc}.

First we have introduced a new compact formulation for \textsc{pmc} based on machine representatives that eliminates the symmetries appearing in the solution space of the intuitive assignment formulation.
For this model, we have proved a lifting lemma that allows us to transform every facet of the \textsc{vcp} polytope into a facet of the \textsc{pmc} polytope.
Using this lemma, we have proved that the \textsc{pmc} polytope is full-dimensional and shown that the inequalities in the model induce facets.
Furthermore, we prove that every facet-defining inequality of the stable set polytope associated with the anti-neighborhood of any vertex also defines a facet of the \textsc{pmc} polytope.
This implies that the external inequalities, a generalization of the Chvátal rank inequalities, are valid and induce facets under weak hypotheses.
To the best of our knowledge, this is the first polyhedral study of the parallel machine scheduling with conflicts.

On the computational side, we have implemented  branch-and-cut algorithms for \textsc{pmc} based on the representative model that apply the cuts given by subclasses of the external inequalities induced by cliques and odd cycles. 
The computational experiments on the hardest instances of the benchmark for \textsc{pmc} show that the proposed algorithms are in general superior (either in running time or quality of the solutions) to the current state-of-the-art solving methods for the problem.
Considering the entire set of instances in the benchmark dataset, the branch-and-cut approach produces average gaps that are up to 10 times smaller for the instances containing graphs of density at least $0.3$, but it has a larger average running time when compared mainly to the solving method due to~\cite{kowalczyk2017exact}.
We remark that the lower average running time of their method comes at the cost of a very intricate implementation that is based on a binary search, iteratively solving a subproblem that is essentially a bin packing problem with conflicts using a branch-and-price algorithm, and incorporating a diving search heuristic to enable an exploration of the solution space using the concept of limited discrepancy search.

We have noticed that most of the benchmark instances of \pmc\ (with at least 50 vertices) have their optimal value very close to the trivial lower bound, which is given by the sum of all processing times divided by the number of machines. 
We have introduced a new class of instances where this difference can become large, and we have conducted  
computational experiments that indicate that these instances become harder to solve as the difference between the optimal value and the trivial lower bound increases.
In this case, the proposed branch-and-cut algorithms clearly outperform the binary search devised in~\cite{kowalczyk2017exact}.

We conclude therefore that the proposed branch-and-cut approach to  \textsc{pmc} yields a simpler algorithm  that is competitive with the state-of-the-art solving methods for \pmc\, especially on instances with conflict graphs of higher density.
Moreover, it is easily amenable to enhancements using new cuts and separation algorithms inherited from the vertex coloring and stable set problems.

In addition to devising new strong inequalities and efficient separation procedures for the branch-and-cut approach, an interesting (but likely laborious) direction for further research is to use the set-covering formulation obtained from a Dantzig-Wolfe decomposition of the representatives model to develop a branch-and-price algorithm for the problem. 
Differently from~\cite{BIANCHESSI2021105464}, this set-covering model would have no symmetries due to permutations of stable sets, and so would not need the auxiliary variables present in Bianchessi and Tresoldi's model to reduce symmetry by ordering the machines by processing times.

\bibliographystyle{apacite}
\bibliography{biblio}

\end{document}